\documentclass[11pt]{amsart}
\usepackage{graphicx}
\usepackage{amsmath}
\usepackage{amssymb}

\usepackage{fullpage}

\newtheorem{theorem}{Theorem}[section]

\newtheorem{proposition}[theorem]{Proposition}
\theoremstyle{definition}
\newtheorem{definition}[theorem]{Definition}
\theoremstyle{remark}

\theoremstyle{definition}

\numberwithin{equation}{section}

\newcommand{\set}[1]{\left\{#1\right\}}

\newcommand{\R}{\mathbb R}
\newcommand{\Z}{\mathbb Z}
\newcommand{\N}{\mathbb N}

\newcommand{\EE}{\mathcal{E}}
\newcommand{\CC}{\mathcal{C}}
\newcommand{\eps}{\varepsilon}
\begin{document}

\title{A complexity approach to the soliton resolution conjecture}%
\author{Claudio Bonanno} 
\address{Dipartimento di Matematica, Universit\`a di Pisa, Largo Bruno Pontecorvo n. 5, 56127 Pisa, Italy. Tel.: +39-050-2213883}
\email{bonanno@dm.unipi.it}
\thanks{I thank Vieri Benci for many stimulating discussions about complexity of solitons. I am partially supported by ``Gruppo Nazionale per l'Analisi Matematica, la Probabilit\`a e le loro Applicazioni (GNAMPA)'' of Istituto Nazionale di Alta Matematica (INdAM), Italy.}
\maketitle

\begin{abstract}
The soliton resolution conjecture is one of the most interesting open problems in the theory of nonlinear dispersive equations. Roughly speaking it asserts that a solution with generic initial condition converges to a finite number of solitons plus a radiative term. In this paper we use the complexity of a finite object, a notion introduced in Algorithmic Information Theory, to show that the soliton resolution conjecture is equivalent to the analogous of the second law of thermodynamics for the complexity of a solution of a dispersive equation.
\end{abstract}

\section{Introduction} \label{sec:intro}

One of the most interesting phenomena observed for solutions of nonlinear dispersive PDEs is described by the \emph{soliton resolution conjecture}. Solitons are solutions of nonlinear equations which are localized, maintaining the same form for all time, and are stable under small perturbations. The resolution conjecture is about a much more strong property of solitons and is based on extensive numerics. It is believed that in many dispersive equations, a ``generic" solution eventually resolves into a finite number of solitons plus a radiative term (see for example \cite{tao1}). This conjecture is vaguely defined and is quite difficult to be more precise due to the ``genericity" of the result, even when restricting to a specific equation. 

There are rigorous results for the Korteweg-de Vries equation and the 1d cubic nonlinear Schr\"odinger equation (NLS), due to the inverse scattering method, and more recent results on the Wave Equation and the high-dimensional NLS based on outstanding and powerful methods in the study of dispersive equations \cite{dkm1,dkm2,dkm3,tao-att}. A probabilistic approach to the resolution conjecture has been introduced for the mass-subcritical NLS in the papers \cite{chatter,ck}. Starting from a discrete NLS it is proved that in the limit of the discretization step going to zero, any uniformly random initial condition has solution converging to a soliton (see \cite[Section 3]{chatter} for a discussion of the result). 

In this paper we introduce a new approach to the soliton resolution conjecture, using the notion \emph{complexity} of finite objects. Roughly speaking, the complexity of a finite object is the amount of information which is necessary to describe it. A formal definition was introduced independently by Kolmogorov and Chaitin, and is now called \emph{Algorithmic Information Content} or \emph{Kolmogorov complexity} (see \cite{lv} and Section \ref{sec:compl} below). The same rough definition can be given for the \emph{Shannon entropy} of a symbol of a string produced by a source of information \cite{shannon}. In fact the two notions turn out to be strictly related when applied to an ergodic dynamical system with a probability invariant measure. In this case, the Shannon entropy is nothing but the \emph{metric entropy} of the system (see for example \cite{kh}), and it coincides with the linear rate of increase of the Kolmogorov complexity of almost all orbits of the system \cite{brudno}.

However the complexity approach has a great advantage with respect to the metric entropy approach, and it is particularly important for applications. Indeed the complexity of an object does not depend on the existence of a mathematical structure. In applications one can think of measuring the complexity of a time series without any information about a mathematical model of the system producing the series. So we can talk of the complexity of the orbit of a system without referring to the phase space or the invariant measure. All we know is that, a posteriori, if the correct mathematical structure exists, then the two notions coincide. As we remarked this is important in applications, but we have benefitted from this property also to introduce the notion of complexity for extended dynamical systems in some cases in which it is not known whether it is possible to define the metric entropy of the system in the classical way. This was done in the paper \cite{bon-col} for PDEs with not-necessarily compact attractors. We refer to \cite{bonanno} for a review of the application of the Kolmogorov complexity to dynamical systems and PDEs.

The situation is much more cumbersome in developing a dynamical system approach to nonlinear dispersive equations, for example the NLS. The dynamical system is a Hamiltonian infinite dimensional system, so there is no attractor in the phase space (at least for the energy norm) and it is difficult to construct a probability invariant measure. The latter problem has been intensively studied after the pioneering paper \cite{lebo}, in which the authors constructed Gibbs invariant measures for NLS on bounded domains, the major successive breakthrough being the introduction in \cite {bour1} of what is now called the Bourgain method. This approach has been extended to the infinite volume case by approximations with bounded domains, see for example \cite{bour2} and \cite{rider}. 

However, we don't need the existence of an invariant measure to apply the complexity approach to nonlinear dispersive equations on unbounded domains. In this paper we introduce this approach in the case of the NLS, and show the connections with the soliton resolution conjecture. In particular we consider the NLS in one dimension, but we believe that the ideas connecting the complexity of a solution of the NLS and the soliton resolution conjecture work in any dimension, the generalization being only subject to technical issues, and for many other dispersive equations with similar structure. Indeed we explicitly use only the Lagrangian structure of the equation.

Let $\psi(t,x)$ be a solution of the NLS
\[
i \partial_t \psi = - h^{-2} \triangle \psi - |\psi|^{2\sigma}\psi\, , \qquad (t,x)\in \R \times \R
\]
in the functional space $\mathcal{F}$. Then from the point of view of a dynamical system, the solution is a map $U:\R \times {\mathcal{F}} \to {\mathcal{F}}$ with $U(t, \psi(0,x)) = \psi(t,x)$. Under suitable assumptions for the well-posedness of the evolution problem, the NLS has ten integral of motion, due to the ten dimension group of symmetries of the Lagrangian associated to the equation, and one integral of motion, the \emph{charge}, due to the so-called gauge invariance. Particularly important for our aims are the energy and the charge.

Given a solution $\psi(t,x)$ with finite energy and charge, we introduce a complexity of the function $\psi(t,\cdot) \in {\mathcal{F}}$. First of all we have to reduce the function $\psi(t,\cdot)$ to a finite object. This is done in more steps: to discretize the space, so to consider only the values $\psi(t, h \ell)$ for $\ell\in \Z$; to look at the modulus $|\psi(t,h\ell)|$, which is the important quantity in identifying a concentrated solutions as a soliton; to introduce a coarse graining for the possible values of $|\psi(t,h\ell)|$. These steps are enough to reduce the function $\psi(t,\cdot) \in {\mathcal{F}}$ to a finite object for which we can measure the complexity. Then we are ready to show that

\vskip 0.2cm
\noindent \textbf{Main result.} \emph{The soliton resolution conjecture is equivalent to an increase of the complexity of the function $\psi(t,\cdot)$ with time.}
\vskip 0.2cm

The formal statement is Theorem \ref{main-corpo} in Section \ref{sec:macro} and contains many more details. Roughly speaking we have to consider the case of big enough solitons and small enough graining. Under these assumptions, we prove that the maximum of the complexity is achieved by functions $\psi(t,\cdot)$ with only one bounded region of the space in which they are big enough, the concentration region, and functions regular in this region, that is not oscillating too much. In this respect, interpreting the complexity as an analogous of the Boltzmann entropy, if we believe in the existence of a law stating the increase of the complexity of a solution of a nonlinear dispersive equation, then our main result implies that we should expect this solution to have in the limit a profile given by a concentrated bump plus small oscillating waves outside the concentration region. That is the solution converges towards a soliton plus a radiative term.

In Section \ref{sec:setting} we introduce the precise setting in which we work. In particular we consider a fixed discretization parameter $h$ of the space $\R$, so we directly introduce the problem using the discrete NLS in 1d. However we consider a general nonlinear term, and not a power term. In Section \ref{sec:compl} we define the complexity for functions $\psi(t,\cdot)$, and in Section \ref{sec:macro} we prove our main theorem.

\section{The setting} \label{sec:setting}

We consider the 1d discrete nonlinear Schr\"odinger equation on $\Z$ with discretization parameter $h>0$ that is
\begin{equation}\label{1-dnls}
i\, \dot \psi_{\ell}(t) = - h^{-2} (\delta^{2} \vec{\psi})_{\ell} - f(|\psi_{\ell}|^{2}) \psi_{\ell} \, , \qquad \ell \in \Z
\end{equation}
where $t\in \R$, the function $\vec{\psi}(t) := \set{\psi_{\ell}(t)}$ has time-dependent components and
\[
(\delta^{2} \vec{\psi})_{\ell} := \psi_{\ell-1}+\psi_{\ell+1}-2\psi_{\ell}\, .
\]
Moreover we assume that there exists $F:\R^{+}\to \R$ of class $C^{2}$ such that $F'(s) = f(s)$ and
\begin{itemize}
\item[\textbf{(F1)}] $F(0)=0$ and $F(s) = o(s)$ as $s\to 0^{+}$; 
\item[\textbf{(F2)}] there exists $s_{0}>0$ such that $F(s)\le 0$ for $s\in (0,s_{0})$ and $F(s)>0$ for $s> s_{0}$;
\item[\textbf{(F3)}] there exists $s_1<s_0$ such that $F$ is non-increasing on $(0,s_1)$, and $F$ is non-decreasing on $(s_{0},+\infty)$.
\end{itemize}
We remark that assumption (F2) excludes the possibility for equation \eqref{1-dnls} to have small solitons. In particular by (F2) it follows that all solitons reach height $s_0$ (see Section \ref{sec:hylo}). This is useful to have a neat distinction between the soliton part and the radiation part in a solution to \eqref{1-dnls}.

By standard methods one can check that for any $\vec{\psi}(0) \in l^{2}(\Z)$, there exists a unique global solution $\vec{\psi} \in C^{1}(\R, l^{2}(\Z))$ to \eqref{1-dnls} and that the following quantities are integrals of motion: the \emph{energy}
\begin{equation}\label{def-energy}
\EE(\vec{\psi}) := \frac{1}{h^{2}} \, \sum_{\ell\in \Z} |\psi_{\ell}-\psi_{\ell-1}|^{2} - \sum_{\ell\in \Z} F(|\psi_{\ell}|^{2})
\end{equation}
and the \emph{charge}
\begin{equation}\label{def-charge}
\CC(\vec{\psi}) := \sum_{\ell\in \Z} |\psi_{\ell}|^{2}\, .
\end{equation}
One useful step to understand the dynamical properties of solutions of dispersive equations is to write the function $\vec{\psi}$ in polar form, namely
\[
\vec{\psi} = (\vec{u},\vec{\theta})\, , 
\]
where $\vec{u}(t) = \set{u_{\ell}(t)}$ and $\vec{\theta}(t) = \set{\theta_{\ell}(t)}$, with 
\[
\psi_{\ell}(t) = u_{\ell}(t) \, e^{i \theta_{\ell}(t)}\, , \qquad \ell\in \Z\, .
\]
Of particular interest are periodic-in-time solutions for which $\theta_{\ell}(t) = \omega t$ for all $\ell$, for some $\omega \in \R$. These solutions are called \emph{standing waves} and sometimes \emph{discrete breathers}.

Using the polar form, we re-write energy as
\[
\EE(\vec{u},\vec{\theta}) = J(\vec{u}) + K(\vec{u},\vec{\theta}) 
\]
with
\begin{equation}\label{def-energy-int}
J(\vec{u}) := \frac{1}{h^{2}} \, \sum_{\ell\in \Z}  |u_{\ell}-u_{\ell-1}|^{2} - \sum_{\ell\in \Z} F(u_{\ell}^{2})
\end{equation}
\begin{equation}\label{def-energy-cin}
K(\vec{u},\vec{\theta}) := h^{-2} \sum_{\ell\in \Z} \, u_{\ell}^{2}\, |\theta_{\ell}-\theta_{\ell-1}|^{2}
\end{equation}
where we have used the approximation $\sin(\theta_{\ell}-\theta_{\ell-1})\sim (\theta_{\ell}-\theta_{\ell-1})$ and $\cos(\theta_{\ell}-\theta_{\ell-1})\sim 1$, which is necessary for the energy to stay finite as $h\to 0^{+}$. Notice that this approximation is an identity for standing waves. The term $J(\vec{u})$ is called the \emph{internal energy} of $\vec{\psi}$ and only depends on the ``shape'' of the functions, and $K(\vec{u},\vec{\theta})$ is called the \emph{kinetic energy}. For a discussion on these quantities we refer to \cite{bgm1}, where this interpretation of the different terms of the energy turns out to be fundamental for the variational approach to the dynamics of soliton solutions.

In the same way, we re-write charge as
\begin{equation}\label{def-charge-pf}
\CC(\vec{u}) := \sum_{\ell\in \Z} u_{\ell}^{2}\, .
\end{equation}

Here we restrict our attention to the study of the internal energy $J$ and the ``shape'' $\vec{u}$ of a function. This is an interesting problem, for example dynamically stable solutions to \eqref{1-dnls} are found as minimizers of $J$ on the manifold of functions $\vec{u}$ with fixed charge (see \cite{weinstein} and references therein). 

For a fixed $\sigma>0$, let
\begin{equation}\label{minimo}
m_{\sigma}:= \inf \set{J(\vec{u})\, :\, \CC(\vec{u}) = \sigma^{2}}\, ,
\end{equation}
and for $m\ge m_{\sigma}$ let
\begin{equation}\label{sottinsieme}
S(m,\sigma):= \set{\vec{u}\in l^{2}(\Z)\, :\, J(\vec{u}) \le m\, , \ \CC(\vec{u}) = \sigma^{2}}\, .
\end{equation}
If $\vec{\psi}(t)$ is a solution to \eqref{1-dnls} with $\vec{\psi}(0)= \vec{\psi}_{0}$, $\EE(\vec{\psi}_{0}) = m$ and $\CC(\vec{\psi}_{0}) = \sigma^{2}$, then by conservation of energy and charge
\[
J(\vec{u}(t)) = \EE(\vec{\psi}(t)) - K(\vec{u}(t),\vec{\theta}(t)) \le \EE(\vec{\psi}(t)) = \EE(\vec{\psi}_{0}) = m
\]
for all $t\in \R$, hence $\vec{u}(t) \in S(m,\sigma)$ for all $t$.

\subsection{Hylomorphic functions} \label{sec:hylo}
We have just recalled that the study of the existence of the minimum for the functional $J$ on the manifold of functions $\vec{u}$ with fixed charge $\CC$ is the procedure to show the existence of \emph{solitons}, namely solitary waves which are orbitally stable. This approach is common to many so-called \emph{focusing dispersive equations} (see \cite{tao-libro} for an introduction) and can be dated back to the paper \cite{coleman} by Coleman et al. More recently, it has been studied in \cite{hylo} for the nonlinear Klein-Gordon equation, considering in particular the role played by the nonlinear term in the properties of the solitons. In the same paper it has been introduced the term \emph{hylomorphic} for this kind of solitons, putting together the greek words ``hyle'' and ``morphe'' which mean respectively ``matter'' and ``form''. 

Using the notation above, we can summarize the ideas underlying the results in \cite{hylo} by saying that solitons with fixed charge $\sigma^{2}$ exist if $m_{\sigma}$ in \eqref{minimo} is smaller than the value of the ``vanishing functions'', namely $\vec{u}^{\eps}$ with $|u^{\eps}_{\ell}|\le \eps$ for all $\ell$, where $\eps\ll 1$. In the particular case of equation \eqref{1-dnls}, considering the functions $\vec{u}^{\eps}$ given by
\begin{equation}\label{frittatine}
u_{\ell}^{\eps} = \left\{
\begin{array}{ll}
\eps\, , & |\ell| \le n\\[0.2cm]
0\, , & |\ell|> n
\end{array}
\right.
\end{equation}
with $\eps= \frac{\sigma}{\sqrt{2n+1}}$, we find
\[
\CC( \vec{u}^{\eps}) = \sum_{\ell=-n}^{n}\, \eps^{2} = (2n+1) \eps^{2} = \sigma^{2}
\]
\[
J( \vec{u}^{\eps}) = \frac{2}{h^{2}} \eps^{2} - \sum_{\ell=-n}^{n}\, F(\eps^{2}) = \frac{2}{h^{2}} \eps^{2} - \sigma^{2}\, \frac{F(\eps^{2})}{\eps^{2}} \to 0^{+} \quad \text{as $\eps \to 0^{+}$}
\] 
by assumption (F1). Hence solitons with fixed charge $\sigma^{2}$ exist if $m_{\sigma}<0$ (cfr. \cite{weinstein}). The existence of such $\sigma$ can be proved by using the functions $\vec{u}^{s}$ given by
\[
u_{\ell}^{\eps} = \left\{
\begin{array}{ll}
\sqrt{s}\, , & |\ell| \le n\\[0.2cm]
0\, , & |\ell|> n
\end{array}
\right.
\]
with $s>s_{0}$ and $\eps= \frac{\sigma}{\sqrt{2n+1}}$, for which we find
\[
\CC( \vec{u}^{s}) = \sum_{\ell=-n}^{n}\, s = (2n+1) s
\]
\[
J( \vec{u}^{s}) = \frac{2}{h^{2}} s - \sum_{\ell=-n}^{n}\, F(s) = \frac{2}{h^{2}} s - (2n+1) F(s)\, .
\] 
Indeed for $n$ big enough we find $J( \vec{u}^{s})<0$ since $F(s)>0$ by (F2). So for large enough charges $\sigma^{2}$ we find $m_{\sigma}<0$. 

We now distinguish the indices $\ell\in \Z$ according to the value $u_{\ell}$, in particular if $J(\vec{u})<0$ then $u_{\ell}>s_{0}$ for some $\ell$ according to (F2). We introduce the notation
\begin{equation}\label{def-region}
U^{-} := \set{ \ell \, :\, u_{\ell}^{2}>s_{0}} \quad \text{and} \quad U^{+} := \set{ \ell \, :\, u_{\ell}^{2}\le s_{0}}\, ,
\end{equation}
which is justified by the fact that $J(\vec{u})<0$ implies $U^{-}\not= \emptyset$.

Following the previous argument we say
\begin{definition}\label{def-hylom-f}
A function $\vec{u}\in l^{2}(\Z)$ is called \emph{hylomorphic} if $U^{-}$ is not empty.
\end{definition}

In the following we consider a fixed value $\sigma^{2}$ for which $m_{\sigma}<0$. For what we have discussed above, given $\vec{\psi}_{0}$ with $\EE(\vec{\psi}_{0})< 0$ and $\CC(\vec{\psi}_{0})=\sigma^{2}$, the solution $\vec{\psi}(t)=(\vec{u}(t),\vec{\theta}(t))$ of \eqref{1-dnls} with initial condition $\vec{\psi}_{0}$ satisfies $\vec{u}(t) \in S(0,\sigma)$ for all $t\in \R$, so that $\vec{u}(t)$ is hylomorphic for all $t\in \R$. 

\section{The complexity approach}\label{sec:compl}
Let $\omega$ be a finite string with characters from a finite alphabet ${\mathcal{A}}$, we use the notation $\omega \in {\mathcal{A}}^{*}$. By \emph{complexity} of a finite object we mean the measure of its \emph{information content}, loosely speaking the minimum amount of bits needed to completely describe the object on a personal computer. So that we consider the complexity as a function
\[
K : {\mathcal{A}}^{*} \to \N
\]
This vague definition can be made rigorous using the concept of \emph{universal Turing machine}, which formalizes the idea of a personal computer and of a programming language used by the computer. In this way we obtain the definition of the \emph{Algorithmic Information Content (AIC)} or \emph{Kolmogorov complexity} (see \cite{lv} for definition, properties and applications of the AIC). In this paper we consider a complexity function $K$ with the following properties
\begin{itemize}
\item[\textbf{(K1)}] there exists a constant $c>0$ such that for all $\omega \in {\mathcal{A}}^{*}$, it holds
\[
K(\omega) \le |\omega| \, \log_{2}(\#({\mathcal{A}})) + c
\]
where $|\omega|$ is the length of $\omega$, and $\#({\mathcal{A}})$ is the cardinality of ${\mathcal{A}}$. In particular using the standard coding of natural numbers by binary words, it holds that for all $m\in \N$
\[
K(m) \le \log_2 (m+1) + c\, ;
\]
\item[\textbf{(K2)}] let $\omega'$ be a sub-string of $\omega$ and $b\in \N^*$ be the list of the positions of the symbols of $\omega$ dropped in $\omega'$, then
\[
K(\omega') \le K(\omega) + \sum_{i=1}^{|b|}\, (\log_2(b_i+1) + c)\, .
\]
In the same way there exists a constant $c'>0$, independent on $\omega$, such that if $\omega''$ is the complement of $\omega'$ in $\omega$, then
\[
K(\omega) \le K(\omega') + K(\omega'') + \sum_{i=1}^{|b|}\, (\log_2(b_i+1) + c) +c'\, ;
\]
\item[\textbf{(K3)}] for each $\omega\in {\mathcal{A}}^{k}$ and each $\omega'\in {\mathcal{A}}^{n}$ it holds
\[
K(\omega\, \omega') \ge K(\omega) + 1
\]
where $\omega\, \omega'$ is the unmarked concatenation of $\omega$ and $\omega'$.
\end{itemize}

We now want to measure the complexity of the functions $\vec{u} \in l^2(\Z)$ with $\CC(\vec{u})=\sigma^{2}$. The problem is that $\vec{u}$ is an infinite string with real components. So the first step is to use a coarse graining description of $\vec{u}$.

Let ${\mathcal{A}}=\set{1,\dots,N}$ and ${\mathcal{P}}=\set{P_{1}, \dots, P_{N}}$ be the partition of the interval $[0,\sigma]$ given by 
\[
P_{k}= \Big[ \frac \sigma N\, (k-1),\, \frac \sigma N\, k \Big)\, , \quad k=1,\dots, N-1\, , \quad P_{N} = \Big[ \frac \sigma N\, (N-1),\, \sigma \Big]\, .
\]
If $\CC(\vec{u})=\sigma^{2}$ then $u_{\ell}\in [0,\sigma]$ for all $\ell$, hence $|u_{\ell}-u_{\ell-1}| \in [0,\sigma]$ for all $\ell$. Moreover there are at most $N^2$ components of $\vec{u}$ bigger than $\frac \sigma N$, indeed
\begin{equation} \label{stima-s}
\# \set{\ell\, :\, u_\ell \ge \frac \sigma N} \, \frac{\sigma^2}{N^2} \le \sum_{\ell \in \Z} u_\ell^2 = \sigma^2\, .
\end{equation}
So we can define a function
\begin{equation}\label{coding}
l^2(\Z) \ni \vec{u} \to \omega(\vec{u}) \in \{{\mathcal{A}},+1,-1\}^{\Z}
\end{equation}
by writing
\[
\omega(\vec{u}) = (\dots s_{-n} \omega_{-n}\, \dots s_{-1} \omega_{-1}\, s_{0} \omega_{0}\, s_{1}\omega_{1}\, s_{2}\omega_{2}\, \dots s_{n-1}\omega_{n-1}\dots )
\]
with $\omega_{i}\in {\mathcal{A}}$ for all $i\in \Z$, given by
\begin{equation}\label{def-omega}
\omega_{i}= k \quad \text{if and only if} \quad |u_{i}-u_{i-1}| \in P_{k}\, ,
\end{equation}
and $s_{i}\in \{\lambda, +1,-1\}$ for all $i\in \Z$, where $\lambda$ denotes the empty string, given by
\begin{equation}\label{def-s}
s_{i} = \left\{ 
\begin{array}{ll} 
+1\, , & \quad \text{if $u_{i}-u_{i-1}>0$ and $\omega_{i}\ge 2$}\\[0.2cm]
-1\, , & \quad \text{if $u_{i}-u_{i-1}<0$ and $\omega_{i}\ge 2$}\\[0.2cm] 
\lambda\, , & \quad \text{if $\omega_{i}= 1$}
\end{array} 
\right. 
\end{equation}
The map \eqref{coding} is a coding of the information contained in the function $\vec{u}$, with a coarsening given by the approximation of the differences $|u_{\ell}-u_{\ell-1}|$ with the integers $\omega_{\ell}$. Hence we approximate the information content of $\vec{u}$ with that of $\omega(\vec{u})$, which is a discrete infinite object. Moreover, by the inequality
\[
\frac \sigma N (\omega_{\ell}-1) \le | u_\ell - u_{\ell-1} | \le \frac \sigma N \omega_{\ell}
\]
it follows that definitively $\omega_i=1$ in both directions, otherwise $J(\vec{u})$ is not finite, and by \eqref{stima-s} definitively $s_i=\lambda$. 

So for each $\vec{u}$ there exists $k=k(N,\vec{u})\in \N$ such that we can restrict our attention to $k$ symbols of  $\omega(\vec{u})$, so that we can consider
\begin{equation}\label{esiste-k}
\omega(\vec{u}) \in \{{\mathcal{A}},+1,-1\}^{k(N,\vec{u})}\, .
\end{equation}
We are then reduced to the study of the complexity of a finite object. 

This is a standard approach to the definition of complexity for orbits of a dynamical system (see \cite{brudno}), the next step being the increasing in the number of sets in the partition $\mathcal{P}$ by letting $N\to \infty$ and studying the asymptotic behavior of the complexity. This procedure is obviously suggested by the definition of metric entropy in dynamical systems (see for example \cite{kh}). Then let us argument on the behavior of the coding $\omega(\vec{u})$ as $N$ increases. As $N$ increases, more and more $\omega_{i}$ for which $u_{i}\not= u_{i-1}$ become greater than 1, and as $N$ diverges they increase linearly with $N$. On the other had, all $s_{i}$ remain constant as soon as $\omega_{i}\ge 2$. Hence using the coding \eqref{coding}, we are able to distinguish two different parts in the information content of the approximated $\vec{u}$: one part, the $\omega_{i}$'s, which is dependent on the coarsening, and one part, the $s_{i}$'s which only depends on $\vec{u}$, and can be considered the description of the ``structure'' of $\vec{u}$. 

Hence for $\vec{u}\in l^2(\Z)$ and $N$ fixed, we use the complexity function $K$ with properties (K1)-(K3) as above, to define
\begin{equation}\label{def-i}
{\mathcal{I}}_{N}(\vec{u}):= K(s_{\vec{u}}) \in \N
\end{equation}
where $s_{\vec{u}} = (s_1,\dots,s_k)$, and the $s_{i}$'s are given by \eqref{def-s}. By the translation invariance of our problem we can always assume that the non-empty symbols $s_i$ have indices in the set $\{1,\dots,k\}$ where $k=k(N,\vec{u})$ as defined above.

\subsection{The macrostates} \label{sec:macro}
In our approach to the soliton resolution conjecture, we now identify the subsets of the functions in $l^2(\Z)$ which play the role of solitons.

\begin{definition}\label{the-bumps}
Let $\vec{u}$ be a hylomorphic function in $S(m,\sigma)$ for $m\in (m_\sigma,0)$. We say that $\vec{u}$ \emph{has a single bump at height $\alpha$}, and write $\vec{u} \in {\mathcal{M}}_{b}^\alpha$, if there is exactly one connected component $U_\alpha^-$ of the set $U^{-}$ on which the maximum of the $u_\ell$ is greater or equal than $\alpha$. We denote by ${\mathcal{M}}_{mb}^\alpha$ the complementary set ${\mathcal{M}}_{mb}^\alpha := S(m,\sigma) \setminus {\mathcal{M}}_{b}^\alpha$.

We say that a vector $\vec{u}$ has a \emph{single regular bump at height $\alpha$ with precision $\beta$}, and write $\vec{u} \in {\mathcal{M}}_{b,r}^{\alpha,\beta}$, if $\vec{u}$ is in ${\mathcal{M}}_{b}^\alpha$ and if there exists $\ell_{0}\in U^{-}_\alpha$ such that $u_{\ell}\le u_{\ell'} +\beta$ for all $\ell,\ell' \in U^{-}_\alpha$ with $\ell<\ell'\le \ell_0$, and $u_{\ell}\ge u_{\ell'} - \beta$ for all $\ell,\ell' \in U^{-}_\alpha$ with $\ell_0 \le \ell< \ell'$. We denote by ${\mathcal{M}}_{b,s}^{\alpha}$ the complementary set ${\mathcal{M}}_{b,s}^{\alpha} := {\mathcal{M}}_{b}^\alpha \setminus {\mathcal{M}}_{b,r}^{\alpha,\beta}$.
\end{definition}

The condition in the definition of a regular bump means that the components $U^{-}$ of a vector in $\vec{u} \in {\mathcal{M}}_{b,r}^{\alpha,\beta}$ are first non-decreasing and then non-increasing up to the precision $\beta$.

We now study the behavior of the information content ${\mathcal{I}}_{N}$ defined in \eqref{def-i} on the different macrostates. We first analyze the contribution of the different components of a vector in $S(m,\sigma)$ to ${\mathcal{I}}_{N}$. Using \eqref{def-region} for a vector $\vec{u} \in S(m,\sigma)$, as in \eqref{stima-s} we get 
\begin{equation}\label{card-meno}
\CC(\vec{u})=\sigma^{2}\quad \text{implies} \quad \# (U^{-})\le  \frac{\sigma^{2}}{s_{0}}\, , 
\end{equation}
where $\#(U^{-})$ is the cardinality of the set. Notice that this estimate is independent on $h$ and $N$. Let us denote by $s^{+}_{\vec{u}}$ and $s^{-}_{\vec{u}}$ the sub-strings of $s_{\vec{u}}=(s_{1}, \dots, s_{k})$ defined in \eqref{def-s}, as
\[
s^{+}_{\vec{u}}:= (s_{\ell})_{\ell\in U^{+}}\quad \text{and} \quad  s^{-}_{\vec{u}}:= (s_{\ell})_{\ell\in U^{-}}\, ,
\]
then by (K1) and (K2)
\[
\begin{array}{c}
K(s^+_{\vec{u}}) \le K(s_{\vec{u}}) + \sum_{\ell\in U^-}\, (\log_2(\ell +1)+c) \le \\[0.2cm]
\le K(s_{\vec{u}}) + \#(U^-) \, (\log_2(k +1)+c) \le \\[0.2cm] 
\le K(s_{\vec{u}}) + \frac{\sigma^{2}}{s_{0}} (\log_2(k +1)+c) +c \, .
\end{array}
\]
where we have used \eqref{card-meno} in the last inequality. Moreover, applying again (K1), (K2) and \eqref{card-meno}
\[
\begin{array}{c}
K(s_{\vec{u}})\le K(s^+_{\vec{u}}) + K(s^-_{\vec{u}}) + \sum_{\ell\in U^-}\, (\log_2(\ell +1)+c) + c' \le \\[0.2cm]
\le K(s^+_{\vec{u}}) + \#(U^-) + c+ \#(U^-) \, (\log_2(k +1)+c) + c' \le \\[0.2cm]
\le K(s^+_{\vec{u}}) + \frac{\sigma^{2}}{s_{0}} (\log_2(k +1)+c+1) +c +c'
\end{array}
\]
So that the difference between the complexity of $s^{+}_{\vec{u}}$ and that of the whole string $s_{\vec{u}}$ is logarithmic in $k$. Then, if $K(s_{\vec{u}})$ is of order $k$ with $k$ large enough, we have $K(s_{\vec{u}}) \approx K(s_{\vec{u}}^+)$. We also introduce the notation
\begin{equation}\label{def-i-pm}
{\mathcal{I}}_{N}^\pm(\vec{u}):= K(s_{\vec{u}}^\pm)\, .
\end{equation}

A second step to study the information content of vectors in the different macrostates is to consider the effect of the energy on ${\mathcal{I}}_{N}$. Let us write $J(\vec{u})$ in \eqref{def-energy} as a sum
\[
J(\vec{u}) = J^{+}(\vec{u}) + J^{-}(\vec{u})
\]
with
\begin{equation}\label{new-j}
J^{\pm}(\vec{u}):= \frac{1}{h^{2}} \, \sum_{\ell\in U^{\pm}}  |u_{\ell}-u_{\ell-1}|^{2}- \sum_{\ell\in U^{\pm}} F(u_{\ell}^{2})
\end{equation}

\begin{proposition}\label{monotonia-k}
Let $\sigma$ and $h$ fixed, and let $\mu>0$ and $N$ such that there exist integers $n,a\ge 1$ satisfying
\begin{equation}\label{cond-an}
\frac{\sigma^2}{N^2} \, a^2 \le s_1\qquad \text{and} \qquad (n+1) \Big( \frac{\sigma^2}{h^2 N^2} \, a^2 + |F(\frac{\sigma^2}{N^2} \, a^2)| \Big) \le \mu\, .
\end{equation}
Then for all $\alpha>s_0$, for each $\vec{v}\in {\mathcal{M}}_{b}^\alpha$ there exists $\vec{u}\in {\mathcal{M}}_{b}^\alpha$ with $J^+(\vec{u})\le J^+(\vec{v})+\mu$, $\CC(\vec{u})\le \CC(\vec{v}) + h^2 \mu$ and $K(s^+_{\vec{u}}) \ge K(s^+_{\vec{v}}) + 1$.
\end{proposition}

\begin{proof}
Given $\vec{v}\in l^2(\Z)$ with $\CC(\vec{v})=\sigma^2$, the coding \eqref{coding} introduces the string $\omega(\vec{v})$ for which there exists $k_1=k_1(N,\vec{v})$, as introduced in \eqref{esiste-k}, such that $(s_{\vec{v}})_\ell \not= \lambda$ only if $\ell\in \set{1,\dots,k_1}$. In particular, for $\ell\not\in \set{1,\dots,k_1}$ we assume
\[
v_\ell < \frac \sigma N\, .
\]
We now construct $\vec{u}\in l^2(\Z)$ with the requested properties. First let
\[
u_\ell = v_\ell \qquad \forall\, \ell \le k_1\, .
\]
For $\ell > k_1$, let $n,a\in \N$ be two integers satisfying \eqref{cond-an} and consider $\tilde s\in \set{-1,+1}^n$ such that
\begin{equation}\label{cond-s}
0\le \sum_{i=1}^j\, \tilde s_i \le a\qquad \forall\, j=1,\dots,n\, .
\end{equation}
We can then define
\[
u_{k_1+j} = \frac \sigma N \, \sum_{i=1}^j\, \tilde s_i\qquad \forall\, j=1,\dots,n\, .
\]
and $u_\ell =0$ for $\ell> k_1+n$.

It follows from \eqref{cond-an} and \eqref{cond-s} that 
\[
u_{k_1+j}^2 \le \frac{\sigma^2}{N^2} \, a^2\le s_1<s_0\qquad \forall\, j=1,\dots,n\, ,
\]
and, using also (F3),
\[
\frac{1}{h^{2}} \, \sum_{\ell>k_1}  |u_{\ell}-u_{\ell-1}|^{2} - \sum_{\ell>k_1} F(u_{\ell}^{2}) \le 
\]
\[
\le (n-1) \frac{\sigma^2}{h^2 N^2} + \frac{\sigma^2}{h^2 N^2} a^2 + \frac{\sigma^2}{h^2 N^2} + n |F(\frac{\sigma^2}{N^2} \, a^2)| \le \mu\, .
\] 
In particular we find that $\vec{u}\in {\mathcal{M}}_{b}^\alpha$ and $\set{\ell \ge k_1} \subset U^+$, moreover
\[
J^+(\vec{u}) \le J^+(\vec{v}) +\mu
\]
and
\[
\CC(\vec{u}) \le \CC(\vec{v}) + n\, \frac{\sigma^2}{N^2} \, a^2 \le \CC(\vec{v}) + h^2 \mu\, .
\]
Finally notice that $\tilde s$ is the coding of $u$ for indices $\ell >k_1$, hence by (K3) we have $K(s^+_{\vec{u}}) \ge K(s^+_{\vec{v}}) +1$.
\qed
\end{proof}

We now show that it is possible to decrease the energy $J$ of a function $\vec{v}\in S(m,\sigma)$ by transforming it from a function with multiple bumps into a function with a single bump. This is a well known fact for the Schr\"odinger equation on $\R^n$, here we obtain a quantitative control on the decrease of the energy.

\begin{proposition}\label{prop-mtob} 
Let $\sigma$ and $h$ fixed, and let $\alpha>s_0$. Then for each $\vec{v} \in {\mathcal{M}}_{mb}^\alpha$  there exists $\vec{u} \in {\mathcal{M}}_b^\alpha$ such that $J(\vec{u}) \le J(\vec{v})-c(\vec{v})$ and $\CC(\vec{u})= \CC(\vec{v})$, where $c(\vec{v})$ is an explicit term depending on $\vec{v}$.
\end{proposition}

\begin{proof}
Let $\vec{v} \in {\mathcal{M}}_{mb}^\alpha$, and let $V^-_1$ and $V^-_2$ two consecutive connected components of $V^-$. Let us introduce the notation $b_i,c_i$, $i=1,2$, for the extreme integers of $V^-_i$, so that as intervals
\[
V^-_1 = [b_1,c_1]\quad \text{and} \quad V^-_1 = [b_2,c_2]\, ,
\]
and $a_i=b_i-1$, $d_i=c_i+1$ are the closest integers in $V^+$. By definition of ${\mathcal{M}}_{mb}^\alpha$, there exists an integer between $V^-_1$ and $V^-_2$, so that $d_1\le a_2 < b_2$.

We now define $\vec{u}$ by
\[
u_\ell = \left\{
\begin{array}{ll}
v_\ell\, , & \text{if $\ell \le a_1$}\\[0.2cm]
v_{a_2+b_1-\ell}\, , & \text{if $b_1 \le \ell \le a_2$}\\[0.2cm]
v_\ell\, , & \text{if $\ell \ge b_2$}
\end{array}
\right.
\]
that is, we take the components of $\vec{v}$ with $\ell \in [b_1,a_2]$ and reverse them, so that now $v_{b_1}$ is next to $v_{b_1}$, merging the two components $V^-_1$ and $V^-_2$ into one.

It follows that $U^-$ has exactly one connected component for $\ell \in [a_1,c_2]$, and if $V^-_1$ and $V^-_2$ were the two only connected components of $V^-$, now $\vec{u} \in {\mathcal{M}}_b^\alpha$. Otherwise we repeat the argument for the other components of $V^-$.

From the definitions, we obtain
\begin{equation}\label{ugua-1}
\sum_{\ell\in \Z} F(v_\ell^2) = \sum_{\ell\in \Z} F(u_\ell^2)\quad \text{and} \quad \sum_{\ell\in \Z} v_\ell^2 = \sum_{\ell\in \Z} u_\ell^2
\end{equation}
and
\[
\sum_{\ell\in \Z} |v_\ell - v_{\ell-1}|^2 - \sum_{\ell\in \Z} |u_\ell - u_{\ell-1}|^2 = 
\]
\[
=|v_{b_1}-v_{a_1}|^2 + |v_{b_2}-v_{a_2}|^2 - |u_{b_1}-u_{a_1}|^2 - |u_{b_2}-u_{a_2}|^2 =
\]
\[
= |v_{b_1}-v_{a_1}|^2 + |v_{b_2}-v_{a_2}|^2 - |v_{a_2}-v_{a_1}|^2 - |v_{b_2}-v_{b_1}|^2\, .
\]
If $v_{a_1}< v_{a_2}$ and $v_{b_1}< v_{b_2}$, and recalling that by definition of $V^-$, we have $v_{b_i}^s> s_0 \ge v_{a_j}^2$, then
\[
|v_{b_1}- v_{a_1}|^2 = |(v_{b_1}- v_{a_2}) + (v_{a_2}-v_{a_1})|^2 \ge |v_{b_1}-v_{a_2}|^2 + |v_{a_2}-v_{a_1}|^2
\]
\[
|v_{b_2}- v_{a_2}|^2 = |(v_{b_2}- v_{b_1}) + (v_{b_1}-v_{a_2})|^2 \ge |v_{b_2}-v_{b_1}|^2 + |v_{b_1}-v_{a_2}|^2 
\]
Then
\[
|v_{b_1}-v_{a_1}|^2 + |v_{b_2}-v_{a_2}|^2 - |v_{a_2}-v_{a_1}|^2 - |v_{b_2}-v_{b_1}|^2 \ge 2\, |v_{b_1}-v_{a_2}|^2\, .
\]
From obvious changes to cover the other possibilities, we find that 
\begin{equation} \label{ugua-2}
\begin{aligned}
& |v_{b_1}-v_{a_1}|^2 + |v_{b_2}-v_{a_2}|^2 - |v_{a_2}-v_{a_1}|^2 - |v_{b_2}-v_{b_1}|^2 \ge \\
& \ge 2\, | \min\{v_{b_1},v_{b_2}\} - \max\{v_{a_1},v_{a_2}\} |^2\, .
\end{aligned}
\end{equation}
Putting together \eqref{ugua-1} and \eqref{ugua-2}, we obtain
\begin{equation}\label{disugua-fin}
\CC(\vec{u}) = \CC(\vec{v}) \quad \text{and} \quad J(\vec{u}) \le J(\vec{v}) - 2\, | \min\{v_{b_1},v_{b_2}\} - \max\{v_{a_1},v_{a_2}\} |^2
\end{equation}
The proof is finished with $c(\vec{v}) = 2\, | \min\{v_{b_1},v_{b_2}\} - \max\{v_{a_1},v_{a_2}\} |^2$.
\qed
\end{proof}

The final step is to analyze the behavior of the energy $J$ for functions with non-regular or regular bumps. In particular we show that for all $\beta>0$, given $\vec{v}\in {\mathcal{M}}_{b,s}^\alpha$ there exists one in ${\mathcal{M}}_{b,r}^{\alpha,\beta}$ with less energy. Again we look for a quantitative estimate on the decrease of energy.

\begin{proposition}\label{prop-stor}
Let $\sigma$ and $h$ fixed, and let $\alpha> s_0$ and $\beta >0$. Then for each $\vec{v} \in {\mathcal{M}}_{b,s}^\alpha$  there exists $\vec{u} \in {\mathcal{M}}_{b,r}^{\alpha,\beta}$ such that $J(\vec{u}) \le J(\vec{v})-\frac{s_0\, \beta^2}{\sigma^2}$ and $\CC(\vec{u})= \CC(\vec{v})$.
\end{proposition}

\begin{proof}
Let $\vec{v} \in {\mathcal{M}}_{b,s}^\alpha$, then there exist $b,c,d \in V^-_\alpha$ with $b<c<d$ such that
\[
v_b \ge v_c + \beta\quad \text{and} \quad v_d \ge v_c + \beta\, .
\]
Without loss of generality we let $v_b \le v_d$ and denote by $a$ the greatest integer smaller than $b$ such that $v_a< v_c$. 

We now define $\vec{u}$ as follows: $u_\ell = v_\ell$ if $\ell \not\in [a,d]$, and for $\ell\in [a,d]$ we rearrange the $v_\ell$ in increasing order. Then again
\begin{equation}\label{ugua-3}
\sum_{\ell\in \Z} F(v_\ell^2) = \sum_{\ell\in \Z} F(u_\ell^2)\quad \text{and} \quad \sum_{\ell\in \Z} v_\ell^2 = \sum_{\ell\in \Z} u_\ell^2
\end{equation}
and 
\[
\sum_{\ell\in \Z} |v_\ell - v_{\ell-1}|^2 - \sum_{\ell\in \Z} |u_\ell - u_{\ell-1}|^2 = 
 \sum_{\ell=a+1}^d |v_\ell - v_{\ell-1}|^2 - \sum_{\ell=a+1}^d |u_\ell - u_{\ell-1}|^2 \, .
\]
Let us first consider the case $[a,d]=\{a,b,c,d\}$, that is there are not intermediate integers between $a,b,c,d$. Then
\[
\sum_{\ell=a+1}^d |v_\ell - v_{\ell-1}|^2 = |v_b-v_a|^2 + |v_c-v_b|^2 + |v_d-v_c|^2\, ,
\]
and rearranging from the smallest to the biggest
\[
\sum_{\ell=a+1}^d |u_\ell - u_{\ell-1}|^2 = |v_c-v_a|^2 + |v_b-v_c|^2 + |v_d-v_b|^2\, .
\]
Using the inequalities
\[
|v_b-v_a|^2 = |(v_b-v_c)+(v_c-v_a)|^2 \ge |v_b-v_c|^2 +  |v_c-v_a|^2
\]
\[
|v_d-v_c|^2 = |(v_d-v_b)+(v_b-v_c)|^2 \ge |v_d-v_b|^2 +  |v_b-v_c|^2
\]
we obtain
\[
\sum_{\ell=a+1}^d |v_\ell - v_{\ell-1}|^2 - \sum_{\ell=a+1}^d |u_\ell - u_{\ell-1}|^2 \ge 2\, |v_b-v_c|^2 \ge 2\beta^2\, .
\]
This inequality can be easily generalized to the case $\#[a,d] = n$, that is $d=a+n-1$, obtaining
\begin{equation}\label{disugua-fin-2}
\sum_{\ell=a+1}^d |v_\ell - v_{\ell-1}|^2 - \sum_{\ell=a+1}^d |u_\ell - u_{\ell-1}|^2 \ge \frac{\beta^2}{n} \ge \frac{s_0\, \beta^2}{\sigma^2}
\end{equation}
where we have used $n\le \#(V^-)$ and \eqref{card-meno}. The proof is finished putting together \eqref{ugua-3} and \eqref{disugua-fin-2}.
\qed
\end{proof}

We are now ready to state the main result. We need to consider how the energy $J^-$ behaves when changing the charge of a function. In particular we define the function $f:(0,1) \to \R^+$ as the function which satisfies the following condition. For each $\vec{u} \in S(m,\sigma)$, let $\vec{u}^\gamma$ be defined by
\begin{equation}\label{dilation}
u_\ell^\gamma = \left\{
\begin{array}{ll}
u_\ell\, , & \text{if $\ell \in U^+$}\\[0.2cm]
\gamma u_\ell\, , & \text{if $\ell \in U^-$}
\end{array}
\right.
\end{equation}
for $\gamma\in (0,1)$. Let $q(\gamma,\vec{u}) := \CC(\vec{u}) - \CC(\vec{u}^\gamma)$ for which
\[
(1-\gamma^2) s_0^2 \le q(\gamma,\vec{u}) \le (1-\gamma^2) \sigma^2\, .
\]
Then $J^+(\vec{u}^\gamma) = J^+(\vec{u})$ and 
\begin{equation}\label{massima-var}
J^-(\vec{u}^\gamma) \le J^-(\vec{u}) + f(\gamma)\, .
\end{equation}
The existence of such a function $f$ follows from the continuity of $J$ and compactness arguments.

\begin{theorem}\label{main-corpo}
Let $\sigma$ and $h$ fixed, and let $\alpha >s_0$ and $\beta\in (0, \alpha-s_0)$. Let $\gamma\in (0,1)$ such that
\[
f(\gamma) \le \min \set{ \frac 12\, \frac{s_0\, \beta^2}{\sigma^2}\, , \, \alpha-\beta-s_0}\, .
\]
Finally we choose $\mu$ such that
\[
\mu \le \min \set{ \frac{(1-\gamma^2)s_0}{h^2}\, , \, \frac 12\, \frac{s_0\, \beta^2}{\sigma^2}\, , \, \alpha-\beta-s_0}\, ,
\]
and $N$ such that there exists integers $n,a\ge 1$ satisfying \eqref{cond-an}.

With these choices, using notation \eqref{def-i-pm}, we have
\[
\max_{{\mathcal{M}}_{b,r}^{\alpha,\beta}} {\mathcal{I}}_{N}^+ >  \max_{{\mathcal{M}}_{mb}^\alpha} {\mathcal{I}}_{N}^+\, ,
\]
and 
\[
\max_{{\mathcal{M}}_{b,r}^{\alpha,\beta}} {\mathcal{I}}_{N}^+ > \max_{{\mathcal{M}}_{b,s}^\alpha} {\mathcal{I}}_{N}^+\, .
\]
\end{theorem}

\begin{proof}
Let $\vec{v} \in {\mathcal{M}}_{mb}^\alpha$ such that 
\[
\max_{{\mathcal{M}}_{mb}^\alpha} {\mathcal{I}}_{N}^+ = K(s^+_{\vec{v}})\, .
\]
By Proposition \ref{prop-mtob} we find $\vec{\tilde w} \in {\mathcal{M}}_{b}^\alpha$ such that
\[
J(\vec{\tilde w}) \le J(\vec{v}) - c(\vec{v})< J(\vec{v})\quad \text{and} \quad \CC(\vec{\tilde w}) = \CC(\vec{v})
\]
where $c(\vec{v})$ is given in \eqref{disugua-fin}.

If $\vec{\tilde w} \in {\mathcal{M}}_{b,r}^{\alpha,\beta}$, by the construction in Proposition \ref{prop-mtob}, we must have that the values of $v_\ell$ for values $\ell$ of boundary for the connected components of $V^-$, are larger than $\alpha-\beta$. Otherwise in $\vec{\tilde w}$ we find two points of maximum at value $\alpha$, and between a point below $\alpha-\beta$. This is in contradiction with $\vec{\tilde w} \in {\mathcal{M}}_{b,r}^{\alpha,\beta}$. Hence in this case we have 
\[
c(\vec{v}) \ge 2\, |\alpha-\beta-s_0|^2\, .
\]

If $\vec{\tilde w} \in {\mathcal{M}}_{b,s}^{\alpha}$, we apply Proposition \ref{prop-stor} and find $\vec{w}\in {\mathcal{M}}_{b,r}^{\alpha,\beta}$ such that
\[
\CC(\vec{w}) = \CC(\vec{\tilde w}) = \CC(\vec{v})\, ,
\]
and
\[
J(\vec{w}) \le J(\vec{\tilde w}) - \frac{s_0\, \beta^2}{\sigma^2} < J(\vec{v}) - \frac{s_0\, \beta^2}{\sigma^2}\, .
\]

In both cases we have found $\vec{w}\in {\mathcal{M}}_{b,r}^{\alpha,\beta}$ such that
\[
\CC(\vec{w}) = \CC(\vec{v})\, ,
\]
and
\[
J(\vec{w}) < J(\vec{v}) - \min\set{ \frac{s_0\, \beta^2}{\sigma^2}\, ,\, 2\, |\alpha-\beta-s_0|^2}\, .
\]

The argument is concluded by application of Proposition \ref{monotonia-k}, but in order to apply the proposition we have first to decrease the charge of the function. So, using \eqref{dilation}, we consider $\vec{w}^\gamma \in {\mathcal{M}}_{b,s}^{\alpha}$, for which
\[
\CC(\vec{w}^\gamma) = \CC(\vec{w}) - q(\gamma,\vec{w}) \le  \CC(\vec{w}) - (1-\gamma^2) s_0^2 \, ,
\]
and by \eqref{massima-var}
\[
J(\vec{w}^\gamma) \le J(\vec{w}) + f(\gamma) \, .
\]
Applying Proposition \ref{monotonia-k} to $\vec{w}^\gamma$, we find $\vec{u}\in {\mathcal{M}}_{b,r}^{\alpha}$ such that
\[
\CC(\vec{u}) \le \CC(\vec{w}^\gamma) + h^2 \mu \le \CC(\vec{v}) - (1-\gamma^2) s_0^2+ h^2 \mu \le \CC(\vec{v})\, ,
\]
and
\[
J(\vec{u}) \le J(\vec{w}) + \mu < J(\vec{v}) - \min\set{ \frac{s_0\, \beta^2}{\sigma^2}\, ,\, 2\, |\alpha-\beta-s_0|^2} + f(\gamma) + \mu \le J(\vec{v})\, .
\]
Finally, notice that in Propositions \ref{prop-mtob} and \ref{prop-stor} we haven't changed the coding of the functions $v$ for what concerns the sets $V^+$. The same is true for the contraction \eqref{dilation}. Hence 
\[
s^+_{\vec{v}} = s^+_{\vec{\vec{w}^\gamma}}\, .
\]
In particular, by Proposition \ref{monotonia-k}
\[
K(s^+_{\vec{u}}) \ge K(s^+_{\vec{v}}) +1 = \max_{{\mathcal{M}}_{mb}^\alpha} {\mathcal{I}}_{N}^+ +1\, .
\]
Finally, if $\CC(\vec{u}) < \CC(\vec{v})$, then we can add a vanishing function defined in \eqref{frittatine} far away from the region $U^-$. In this way we can make $\CC(\vec{u}) = \CC(\vec{v})$ and still $J(\vec{u})\le J(\vec{v})$ by the computations of Section \ref{sec:hylo}. This last step does not affect the complexity if the vanishing functions is small enough, otherwise it could further increase th complexity of $s_{\vec{u}}^+$.
 
Since $\vec{u}\in {\mathcal{M}}_{b,r}^{\alpha}$ this concludes the proof of the first inequality. The second inequality is proved in the same way by starting with $v\in {\mathcal{M}}_{b,s}^{\alpha}$.
\qed
\end{proof}

\section{Conclusion} \label{sec:concl}

We have first divided the manifold of functions $\psi(t,\cdot) \in l^2(\Z)$ with fixed charge into regions, which we called macrostates, described by the behavior of the profiles of the modulus $u=|\psi(t,\cdot)|$ above the level $s_0$. Then we have shown that it is possible to go from one macrostate to the other increasing the complexity of the function $u$. In particular we have shown that the complexity of $u$ is dominated by the complexity of the part of $u$ below $s_0$, and that it is this part of the complexity which increases. 

In the proof of the main result, we have shown how to go from the macrostate of functions with many bumps to the macrostate of functions with only one bump, and then to the macrostate of functions with only one regular bump. Looking at the details, we can also conclude that these transformations decrease the complexity of the part of the functions above $u$. This is not relevant for our main result since this part of the functions is uniformly bounded. However we can interpret this by saying that there is a flow of complexity from the structured part of the function, the soliton part which lies above $s_0$, towards the radiative part, where the function is below $s_0$. This flow of complexity leads to a global increase of complexity and to the formation of one regular structured part. This interpretation is compatible with an analogy between the complexity and the entropy of the system, as regulated by the second law of thermodynamics.

\end{document}